\documentclass[sigconf,9pt]{acmart}

\renewcommand\footnotetextcopyrightpermission[1]{} 

\copyrightyear{2026}
\acmYear{2026}
\setcopyright{acmlicensed}
\acmConference[DAC '26]{63st ACM/IEEE Design
Automation Conference}{July 26-29, 2026}{Long Beach, CA, USA}

\usepackage[ruled,linesnumbered]{algorithm2e}
\usepackage{multirow}
\usepackage{graphicx}
\usepackage{xcolor}
\usepackage{comment}
\usepackage{subfigure}
\usepackage[labelformat=simple]{subcaption}
\usepackage{booktabs}
\usepackage{adjustbox}

\graphicspath{{Figures/}}

\newcommand{\eg}{{\it e.g.}\xspace}
\newcommand{\ie}{{\it i.e.}\xspace}

\begin{document}

\title[Exploiting Dependency and Parallelism: Real-Time Scheduling and Analysis for GPU Tasks]{Exploiting Dependency and Parallelism:\\Real-Time Scheduling and Analysis for GPU Tasks}

\author{Yuanhai Zhang}
\affiliation{
  \institution{Sun Yat-sen University}
  \city{Guang Zhou}
  \country{China}
}
\author{Songyang He}
\affiliation{
  \institution{Sun Yat-sen University}
  \city{Guang Zhou}
  \country{China}
}

\author{Ruizhe Gou}
\affiliation{
  \institution{Hunan University}
  \city{Chang Sha}
  \country{China}
}

\author{Mingyue Cui}
\affiliation{
  \institution{Sun Yat-sen University}
  \city{Guangzhou}
  \country{China}
  }
  
\author{Boyang Li}
\affiliation{
  \institution{Sun Yat-sen University}
  \city{Guangzhou}
  \country{China}
  }
  
\author{Shuai Zhao}
\authornote{Corresponding author. Email: zhaosh56@mail.sysu.edu.cn}
\affiliation{
  \institution{Sun Yat-sen University}
  \city{Guangzhou}
  \country{China}
}

\author{Kai Huang}
\affiliation{
  \institution{Sun Yat-sen University}
  \city{Guangzhou}
  \country{China}
}

\begin{abstract}
With the rapid advancement of Artificial Intelligence, the Graphics Processing Unit (GPU) has become increasingly essential across a growing number of safety-critical application domains.
Applying a GPU is indispensable for parallel computing; however, the complex data dependencies and resource contention across kernels within a GPU task may unpredictably delay its execution time.
To address these problems, this paper presents a scheduling and analysis method for Directed Acyclic Graph (DAG)-structured GPU tasks.
Given a DAG representation, the proposed scheduling scales the kernel-level parallelism and establishes inter-kernel dependencies to provide a reduced and predictable DAG response time.
The corresponding timing analysis yields a safe yet non-pessimistic makespan bound without any assumption on kernel priorities.
The proposed method is implemented using the standard CUDA API, requiring no additional software or hardware support. Experimental results under synthetic and real-world benchmarks
demonstrate that the proposed approach effectively reduces the worst-case makespan and measured task execution time compared to the existing methods up to $32.8\%$ and $21.3\%$, respectively.
\end{abstract}
\maketitle

\section{Introduction}

The rapid advancement of Artificial Intelligence (AI) has made Graphics Processing Units (GPUs) indispensable in many safety-critical domains such as autonomous driving, avionics, and industrial control, where both computational efficiency and timing predictability are crucial.
GPU platforms provide massive parallel execution capabilities and are widely adopted in such systems to accelerate computation.
A GPU task typically consists of multiple computation stages, \ie, kernels in the CUDA framework, which are often represented as a Directed Acyclic Graph (DAG) by modern AI compilers~\cite{li2020deep}.
Although CUDA supports kernel-level parallel execution, the complex data dependencies and heterogeneous workload distributions within GPU tasks often delay the response time, particularly for tasks composed of numerous lightweight kernels with small computation loads~\cite{ma2020rammer}.
Moreover, the black-box nature of GPU hardware execution further increases the pessimism of timing analysis and limits the flexibility of scheduling.

Improving kernel-level concurrency is an effective way to reduce the overall response time of a GPU task, \ie, its makespan~\cite{lin2023kesco}.
Several studies have enhanced concurrency using standard CUDA APIs such as streams and events~\cite{lin2023kesco,parravicini2021dag}, while others have proposed resource allocators to accelerate kernel execution~\cite{STGM,zhang2025improving}. 
However, these approaches do not consider the various computation load and resource requirements across kernels, resulting in unpredictable prolonged kernel execution time~\cite{Zeng2023Enabling}.
Existing response time analysis methods for GPU tasks are also limited in scope.
Most analysis approaches focus on the GPU task with single or sequential kernels ~\cite{STGM,zou2023rtgpu, Ni2025hard,xu2025real}, while others fail to consider the dependency delay between kernels~\cite{yang2017response}.
To the best of our knowledge, the timing analysis for DAG-structured GPU tasks remains unseen.

Although CUDA streams enforce a deterministic kernel submission order, the GPU’s hardware scheduler introduces nondeterminism in the actual execution order of concurrently launched kernels. Moreover, kernel-level preemption through prioritized CUDA streams is unreliable~\cite{Bakita2024Demystifying}, meaning that the DAG scheduling and analysis approaches relying on node-level priorities cannot be applied~\cite{he2019,zhao2020,zhao_tpds}.
Some studies instead schedule GPU tasks via priority mechanisms at the CPU process level~\cite{DNN_real,han2025real}; however, such approaches do not suit tasks that contain multiple kernels with complex dependencies.
Therefore, developing a predictable scheduling and timing analysis framework for GPU tasks that fully exploits kernel dependency and parallelisms remains an open problem.

\textbf{Contribution. }
To address the aforementioned problems, this paper proposes a scheduling and timing analysis framework for DAG-structured GPU tasks, providing a reduced and predictable makespan.
Given a DAG representation of a GPU task, the proposed method decomposes the task into a sequence of balanced groups, each containing kernels that can execute concurrently.
We propose a \textit{parallelism scaling} mechanism for kernels within a group by adjusting the computing resource requirement according to each node’s computation load, achieving balanced execution times and avoiding resource contention.
To further reduce the makespan, parallel nodes of a balanced group are opportunistically launched.
For large kernels that exceed the spare GPU capacity, a \textit{node segmentation} mechanism is constructed to divide them into smaller sequential segments.
The sequential execution order of each balanced group is ensured by the constructed \textit{extra dependency} on the DAG, resulting in a predictable overall makespan.
The proposed method can be entirely implemented within the scope of the standard CUDA API, without requiring any additional hardware or software support.
A large-scale synthesized experiment and a case study demonstrate the effectiveness of the proposed scheduling and timing analysis framework.
In summary, the contributions of this paper are as follows:
\begin{itemize}
\item A scheduling scheme integrating \textit{parallelism scaling}, \textit{node segmentation} , and \textit{extra dependency} mechanism for DAG-structured GPU tasks to reduce the makespan.
\item A timing analysis of the worst-case makespan of the GPU tasks without the assumption of node priorities.
\item Experimental results show that the proposed method outperforms the existing methods by up to $32.8\%$ in worst-case makespan and $21.3\%$ in measured task execution time.
\end{itemize}

\section{System Model}

This section introduces the system model, which describes a GPU platform executing a task in a DAG-structured intermediate representation (IR). 
Figure~\ref{fig:dag_gpu} illustrates a typical GPU task, which begins with data transfer from the CPU to the GPU (\textit{memHtoD}), 
proceeds through a computation phase modeled as a DAG, and concludes with data transfer back to the CPU (\textit{memDtoH}). 
The GPU platform contains $M$ homogeneous Streaming Multiprocessors (SMs), \eg $M=80$ for NVIDIA V100.
Although NVIDIA GPUs are assumed, the proposed method applies to any platform following the Single Instruction Multiple Thread (SIMT) execution model.

\subsection{GPU Task Model}

The proposed method focuses on a single periodic DAG-structured GPU task. The computation stage on the GPU is periodically released by a CPU process in the system.
A periodic DAG task is defined as $\tau = \{G, T, D\}$, where $G=(V, E)$ is a directed acyclic graph, $T$ is the period, and $D=T$ is the implicit deadline.  
Each node $v_i \in V$ represents a kernel, and each edge $e_{i,j}=(v_i,v_j)\in E$ captures a data dependency (e.g., write-after-read), 
ensuring that the successor $v_j$ starts only after the predecessor $v_i$ completes~\cite{parravicini2021dag}.  
A node is said to be \textit{released} when all its predecessors have finished.  
Nodes in a set $S$ without predecessors or successors form the \emph{source} and \emph{sink} node sets, denoted by $src(S)$ and $sink(S)$, respectively. Without loss of generosity, we assume a single source and a single sink node for $G$. 
If $v_i$ is a transitive predecessor of $v_j$, then $v_i$ is an ancestor (\ie, $v_i \in ance(v_j)$).  
Nodes that share no transitive dependencies with $v_i$ form its concurrent set, denoted as $con(v_i)$.  
A path $\lambda = \langle v_1, v_2, \ldots, v_k\rangle$ is defined as a sequence of nodes satisfying $(v_i,v_{i+1}) \in E$ for all $v_i \in \lambda \setminus v_k$. The length of a path is the sum of the computation loads of its nodes.

\subsection{Execution Model}
In the SIMT execution model, parallel execution for each kernel is configured through two launch parameters, $D_g$ and $D_b$, 
which represent the thread grid and block dimensions, respectively. 

\textbf{Kernel execution:}
In the DAG model, each node $v_i = \{\hat{C}_i, C_i, m_i\}$ represents a kernel execute on GPU, 
where $\hat{C}_i$ denotes the computation load, $C_i$ is the kernel execution time, and $m_i$ represents the number of SMs required for execution. 
The parallelism $m_i$ can be adjusted by varying $D_g$ and $D_b$ in the kernel launch configuration. 
Each SM has a finite resource capacity, limited by factors such as register usage and shared memory. 
When the number of active threads exceeds the capacity of a single SM, more SMs are required to accommodate the workload.
For example, if a kernel is launched with $D_g = 1$
and $D_b = 256$, fully occupying one SM. 
Increasing $D_g$ in this case scales up $m_i$, thereby increasing the degree of parallelism.

The execution time of a kernel follows Gustafson's law, 
where the time per thread decreases as the number of threads increases~\cite{zou2023rtgpu}. 
The computation load $\hat{C}_i$ is defined as the maximum kernel execution time when $v_i$ occupies one SM. 
The execution time under parallelism $m_i$ is modeled as
\begin{equation}\label{eq:kernel_execution}
    C_i = \max \left(t_{\min}, \left\lceil \frac{m_i}{M} \right\rceil \frac{\hat{C}_i}{m_i}\right),
\end{equation}
where $t_{\min}$ denotes the theoretical lower bound constrained by memory throughput, 
according to the Roofline performance model~\cite{williams2009roofline}. 
The maximum useful parallelism for accelerating $v_i$ is $m_i^{\max} = \hat{C}_i / t_{\min}$. 
For simplicity, we set $t_{\min} = 1$ as the basic time unit, 
which does not affect the generality of the proposed method.
Note that increasing $m_i$ does not always minimize $C_i$. 
If some SMs are already occupied, the node cannot achieve the desired parallelism and may incur unpredictable delay.

\begin{figure}[t]
    \centering
    \includegraphics[width=1\linewidth]{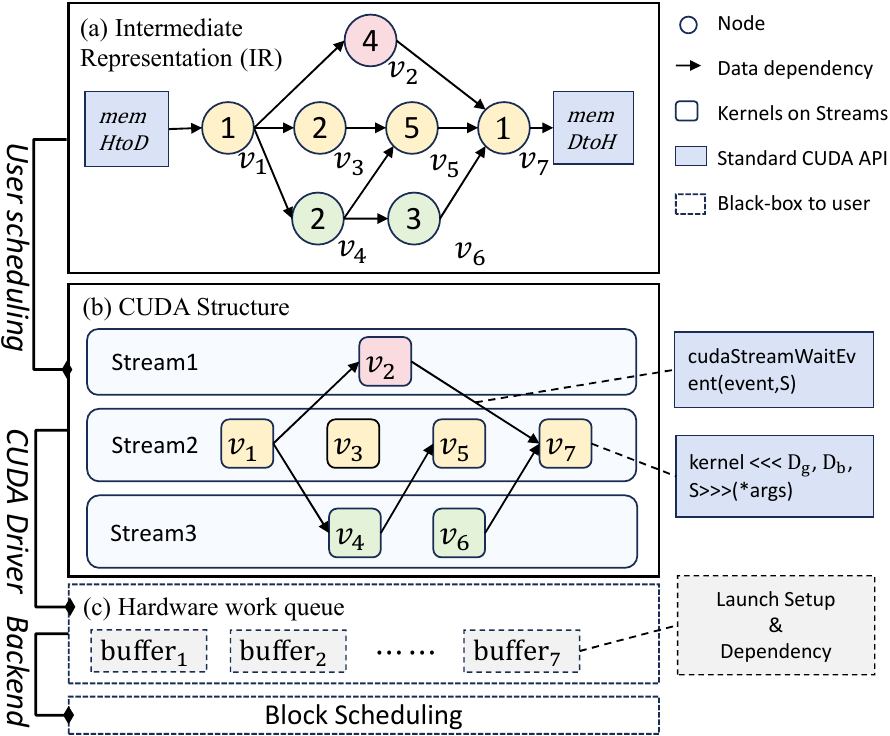}
    \caption{The execution of a GPU task in DAG representation}
    \label{fig:dag_gpu}
\end{figure}

\textbf{Kernel concurrency: }
The standard CUDA API provides kernel concurrency and dependency control through streams and events. Each kernel in the DAG is launched on a stream. Kernels within the same stream execute sequentially, while kernels on different streams may run concurrently if resources allow. Inter-stream dependencies are enforced using CUDA events. 
Kernel dependencies can also be controlled using CUDA Graphs API, which additionally reduces kernel launch overhead, an important advantage for tasks composed of many small kernels~\cite{NVIDIA_GettingStartedCudaGraphs_2020}.
When a GPU task is launched, the CUDA driver generates command buffers encoding kernel configurations and dependencies, and places them into the hardware work queue. The block scheduler dispatches kernels whose dependencies are satisfied to available SMs, while others remain pending~\cite{otterness2017inferring}. Due to the black-box nature of the hardware scheduler, the execution order of concurrent kernels and the resulting contention delays are unpredictable to users.

\textbf{Example:}
Figure~\ref{fig:dag_gpu} illustrates the execution of a DAG-structured GPU task with $7$ nodes, 
$V = \{v_1, \ldots, v_7\}$, where $v_1$ and $v_7$ are the source and sink nodes of $G$, respectively. 
The number inside each node indicates its computation load, \eg, $\hat{C}_3 = 3$ time units. The concurrent node set of $v_2$ is $con(v_2) = \{v_3,v_4\}$. 
The nodes in the same color are launched on the same streams as:  $S_1 = \{v_2\}$, $S_2 = \{v_1, v_3, v_5, v_7\}$, and $S_3 = \{v_4, v_6\}$. 
Inter-kernel dependencies like $\{v_1,v_4\}$ and $\{v_2,v_7\}$ are implemented via CUDA events. 
These kernels are placed into the hardware work queue as command buffers and dispatched to SMs by the block scheduler. 
The actual execution order between concurrent nodes such as $v_2$ and $v_3$, however, remains unpredictable.

\textbf{Problem formulation:}
Given a DAG-structured GPU task, the objective is to minimize the task execution time while guaranteeing a predictable DAG makespan. 
This is achieved by appropriately scaling each node’s parallelism $m_i$ and adjusting the DAG structure to balance concurrent execution under dependency constraints.

\section{Sub-graph Division}\label{sec:division}

To minimize the delay caused by complex data dependencies in the DAG, concurrent nodes without dependency constraints are extracted.
This section presents a sub-graph division method that decomposes the DAG into a sequence of disjoint balanced groups, forming the basis for the proposed scheduling and timing analysis.
First, the nodes that may experience dependency delay are identified, and their ancestors are organized into disjoint blocks.
Then, within each block, concurrent nodes are further arranged into a sequence of balanced groups, the basic sub-graph for scheduling and analysis.

\subsection{Block Formulation}

A node with multiple predecessors is defined as a \textit{join node}, and it cannot start execution until all predecessors finish. 
A join node may experience a dependency delay, which source from the unbalanced finish time of its predecessors.
Join nodes are first identified and sorted according to their \textit{Cumulative ancestor workload}. The notation $W^{anc}(v_i)$ de notes the cumulative computation load cost when $v_i$ finished as:
\begin{equation}
    W^{anc}(v_i) = \sum_{\forall v_j \in ance(v_i) \cup \{v_i\}} \hat{C}_j,  \quad \forall v_i \in G.
\end{equation}

The ordered list $\Theta = \langle \theta_1, \ldots, \theta_N \rangle$ contains all the $N$ join nodes sorted by descending cumulative ancestor workload, which also preserves the DAG’s topological order.

Next, the DAG is partitioned into a sequence of disjoint node blocks $B =[B_1, B_2, \ldots, B_{N+1}]$ (not confused with the thread block in CUDA), where each block $B_k$ includes all ancestors of the join node $\theta_k$ that are not assigned to previous blocks:
\begin{equation}\label{eq:b_k}
B_k = ance(\theta_k) \setminus \bigcup_{j=1}^{k-1} B_j ,\quad \forall \theta_k \in \Theta
\end{equation}
After all the $N$ blocks are formed, the remaining nodes are collected into a residual block $B_{N+1}$. Notice that the blocks in $B$ also follow the order of $\Theta$.

Since all the $\theta_k$ do not belong to their ancestor block $B_k$, all the nodes in $B_k$ contain at most one predecessor in the same block. 
Any $v_i \in B_k$ is a \textit{local source} if all its predecessors lie outside $B_k$, \ie $v_i \in src(B_k)$ , and a \textit{local sink} if all successors lie outside $B_k$, \ie $v_i \in sink(B_k)$. 
A \textit{local complete path} $\lambda_k^j$ connects a local source and a local sink:
\begin{equation}\label{eq:local_complete_path}
\lambda_k^j = \langle v_s, \ldots, v_e \rangle, \quad
v_s \in source(B_k),\ v_e \in sink(B_k).
\end{equation}
Each node in $B_k$ belongs to at least one local complete path, and each path ends with a predecessor of the join node $\theta_k$. 
The notation $\Lambda_k$ denotes the set of local complete paths of $B_k$.

\subsection{Balanced Groups Construction}
To minimize the release time of $\theta_k$, the proposed scheduling balances the finish time of each local complete path in $B_k$.
The concurrent nodes in different local complete paths that can occupy all the $M$ SMs are extracted and organized into a \textit{balanced group}. 
Following the topology order, each group can safely run sequentially and form the ordered balanced group list $\Pi$.
The formulation of a balanced group $\pi_j \in \Pi$ satisfies the following rules:

\textbf{Rule 1:} $|\pi_j| \le M,  \forall \pi_j \in \Pi.$ 
Each kernel requires at least one SM; therefore, the number of concurrent nodes cannot exceed $M$. Otherwise, the resource contention among $\pi_j$ can not be avoided.

\textbf{Rule 2:} If $|\pi_j|>1$, then $m_i^{max}<M,  \forall v_i \in \pi_j.$  
For a large kernel with enough computation load to fully occupy all the SMs ($m_i^{max} \ge M$), the kernel-level parallelism is unnecessary.

\textbf{Algorithm~\ref{alg:garpah_division}} outlines the process of complete sub-graph division in Section \ref{sec:division}, which takes $G$ as the input and outputs the ordered balanced group list $\Pi$. First, the join nodes are detected and ordered by \textit{Cumulative ancestor workload} (Lines 1). Then the blocks are constructed (Lines 2-6). For each block, after the local complete path set is identified (Lines 9), balanced groups are formed by iteratively extracting the head nodes of each local complete path until all nodes are grouped (Lines 10–17), following \textbf{Rule 1} (Line 11) and \textbf{Rule 2} (Lines 12-14).  
The time complexity of Algorithm~\ref {alg:garpah_division} is $\mathcal{O}(|V|)$.

\begin{algorithm}[h]
\caption{\textit{Sub-graph\_Division}($G$)}
\label{alg:garpah_division}
\KwIn{DAG $G = (V,E)$}
\KwOut{Balanced group set $\Pi$}

$\Theta \gets$ join nodes sorted by $W^{anc}$\;
Initialize $B \gets [\,]$\;
\For{$\theta_k \in \Theta$}{
    $B_k \gets$ Eq.~\eqref{eq:b_k}; Append $B_k$ to $B$\; 
}
Append residual block to $B$\; 
Initialize $\Pi \gets [\,]$\;
\For{each $B_k \in B$}{
    Identify $\Lambda_k$ via Eq.~\eqref{eq:local_complete_path}\;
    \While{$\Lambda_k \neq \emptyset$}{
        $\pi \gets Top_M\{\text{head}(\lambda) \mid \lambda \in \Lambda_k\}$ by $W^{anc}$\;
        \If{$m_i^{max} > M , \forall v_i \in \pi$}{$\pi \gets$ node with max $m_i^{max}$ in $\pi$\;}
        Remove nodes in $\pi$ from $\Lambda_k$\;
        Append $\pi$ to $\Pi$\;
    }
}
\vspace{-0.1cm}

\end{algorithm}

\textbf{Example: }The DAG in Figure \ref{fig:dag_gpu} contains $2$ join nodes: $\theta_1 = v_5$ and $\theta_2 = v_7$. The DAG is decomposed into $3$ disjoint blocks: $B_1 = \{v_1,v_3,v_4\}$, $B_2 = \{v_2,v_5,v_6\}$, and $B_3 = \{v_7\}$.
The corresponding local complete path sets are $\Lambda_1 = \{\langle v_1,v_3 \rangle ,\langle v_1,v_4 \rangle \}$ and $\Lambda_2 = \{ \langle v_2 \rangle , \langle v_5 \rangle, \langle v_6 \rangle \}$ . The balanced group list of this DAG is $\Pi = \langle\{v_1\},\{v_3,v_4\},\{v_2,v_5,v_6\},\{v_7\}\rangle$.

\section{Scheduling and Analysis}\label{sec:sche_ana}

This section presents a scheduling and timing analysis framework for DAG-structured GPU tasks on the basis of the proposed sub-graph division. The scheduling method determines the execution order and the required parallelism of each node. Building on the schedule scheme, the timing analysis derives the overall DAG makespan.

\subsection{Scheduling}\label{sec:schedule}

The proposed scheduling framework operates by scaling the required parallelism of each node and restructuring the DAG through node segmentation and extra dependency constraints. Importantly, it does not rely on GPU priority mechanisms and can be fully implemented using the standard CUDA API. The parallelism scaling scheme equalizes the execution times of concurrently executed nodes according to their computation loads, thereby minimizing the completion time of each balanced group. The node segmentation mechanism and the extra dependency further improve hardware utilization while ensuring a predictable overall DAG makespan.

\textbf{Parallelism scaling: }
For lightweight kernels within a balanced group, the scheduler aims to equalize their execution times.
For any balanced group $\pi_j \in \Pi$, the parallelism $m_i$ of each node $v_i \in \pi_j$ is determined in proportion to its computation load:
\begin{equation}\label{eq:m_i}
m_i = \min \left( m_i^{max}, \text{round}\left( \frac{\hat{C_i}}{W(\pi_j)} \cdot M \right) \right),
\quad \forall v_i \in \pi_j
\end{equation}
where $W(\pi_j) = \sum_{v_i \in \pi_j}\hat{C_i}$ is the total computation load of the group.
This configuration ensures that small kernels share GPU resources proportionally, while yielding balanced finishing times.
This parallelism setup ensures the computing resource requirement of a group would not exceed the GPU capacity (\ie $\sum m_i \leq M, \forall v_i \in \pi_j$), which guarantees each node can achieve its required number of SM, thereby achieving the execution time in Eq. \eqref{eq:m_i}.

\textbf{Launching parallel nodes: }
If a balanced group $\pi_j$ does not fully occupy the GPU, parallel nodes may be launched on the spare SMs. The GPU spare capacity when $\pi_j$ executing is:
\begin{equation}\label{eq:Ws}
    WS_j = \max(C_i) \times S_j, \quad \forall v_i \in \pi_j
\end{equation}
where $S_j = M - \sum m_i$ denote the number of remaining SMs.
The set of parallel nodes of a group $\pi_j$ is defined as:
\begin{equation}\label{eq:para_set}
para(\pi_j) =src\left(\bigcup_{v_i \in \pi_j} con(v_i) \setminus \pi_j \right) \cap \text{released\_node}
\end{equation}
where $src(\cdot)$ returns the source nodes of the given set. 
Only nodes that are released at the moment when $\pi_j$ is released are put into the parallel set, and whether a node is released is known during the process of scheduling. After the parallel set is formed, nodes in $para(\pi_j)$ are kept chosen to launch in descending order of cumulative ancestor workload until the spare capacity is exhausted or no parallel nodes remain.
For each chosen node $v_c$, its parallelism is assigned greedily as $m_c = \min(m_c^{max}, S_j)$ to fully utilize the spare SM.

\textbf{Node segmentation:} 
In the SIMT execution model, a node can be divided into multiple segments, each with its own computation load and parallelism configuration, but executing the same function as the original node.
During the execution of $\pi_j$, if a selected node $v_c \in para(\pi_j)$ has a computation load exceeding the available capacity (\ie, $\hat{C}_c > WS_j$), it is split into two segments $v_{c,p}$ and $v_{c,r}$ with
$\hat{C}_{c,p} = WS_j$ and $\hat{C}_{c,r} = \hat{C}_c - \hat{C}_{c,p}$.
The parallel segment $v_{c,p}$ is scaled as the chosen node with $m_{c,p} = \min (m_{c,p}^{max}, S_j)$ and finishes together with $\pi_j$, while the residual segment $v_{c,r}$ inherits the forward dependencies of $v_c$ and execute after $\pi_j$ finished.
Since at most one segmentation occurs per group, the total number of segmentations in a DAG is bounded by the number of balanced groups $|\Pi|$, which is smaller than $|V|$.

\textbf{Extra dependency:}
To preserve the sequential execution order of each balanced group, extra dependencies are introduced.
Let $v_R = \arg\max_{v_i \in \pi_j} C_i$\; denote the node in $\pi_j$ with the maximum execution time.
For each unlaunched node $v_n \in para(\pi_j)$ (including unexecuted segments), an edge $(v_R, v_n)$ is added, forcing these nodes to start only after the current group finishes.
For each launched node or segment $v_n \in para(\pi_j)$, extra dependencies $(v_n, succ(v_R))$ are added to ensure they finish no later than the current group.

\begin{algorithm}[h]
\caption{\textit{Scheduling}($G, \Pi$)}
\label{alg:scheduling}
\KwIn{DAG $G=(V,E)$, balanced group list $\Pi$}
\KwOut{Parallelism $m_i$, segmentation $\overline{V}$, extra deps $\overline{E}$}

Initialize $\overline{V} \gets \emptyset$, $\overline{E} \gets \emptyset$\;
Released\_nodes $\gets \{src(G)\}$ \;
\ForEach{$\pi_j \in \Pi$}{
    $m_i \gets$ Eq.~\eqref{eq:m_i}, \quad $\forall v_i \in \pi_j$\;
    $WS_j \gets$ Eq.~\eqref{eq:Ws} \;
    $para(\pi_j) \leftarrow$ Eq. \eqref{eq:para_set}\;
    \While{$WS_j> 0$ \textbf{and} $para(\pi_j) \neq \emptyset$}{
      $v_c \gets $ Top$(para(\pi_j))$ by $W^{anc}$\;
      \If{$\hat{C_c} > WS_j$}{
        Add $\{v_c:(v_{c,p}, v_{c,r})\} $ to  $\overline{V}$\;
        $\hat{C}_c \gets \hat{C}_{c,p}$\;
      }
    $m_c \gets \min\{m_c^{max},S_j\}$\;
    Add extra dependencies to $\overline{E}$\;
    $WS_j \leftarrow WS_j - \hat{C_c}$\;
    Remove $v_c$ from $para(\pi_j)$\;
    }
    Update Released\_nodes \;
}
 \end{algorithm}
 
\vspace{-0.2cm}

\textbf{Algorithm~\ref{alg:scheduling}} outlines the complete scheduling process of Section \ref{sec:schedule}, which takes $G$ and $\Pi$ as the input and outputs the scheduling scheme including parallelism setup $m_i$ for each node, node segmentation scheme $\overline{V}$, and the extra dependency set $\overline{E}$.
For each balanced group $\pi\in\Pi$, the scheduler first assigns parallelism values for each node in $\pi$ (Line 4). Then, the spare capacity $WS_j$ and parallel nodes $para(\pi_j)$ are initialized (Line 5-6).
The parallel nodes are launched following the descending cumulative ancestor workload order until spare capacity is filled or all the parallel nodes are launched (Line 7-17). 
Node segmentation is applied when the chosen node $v_c$ is too big to fit in $WS_j$ (Line 9-12). The chosen node (segment) is always set with the maximum parallelism (Line 13). After that, the extra dependencies are inserted to ensure the execution order (Line 14), and the status is updated for the next parallel node chosen (Line 15-16). The release nodes are updated when the scheduling of a group finishes (Line 18). 
The time complexity of Algorithm \ref{alg:scheduling} is $\mathcal{O}(|\pi|) \prec \mathcal{O}(|V|)$.

\begin{figure}[htpb!]
    \centering
    \includegraphics[width=0.9\linewidth]{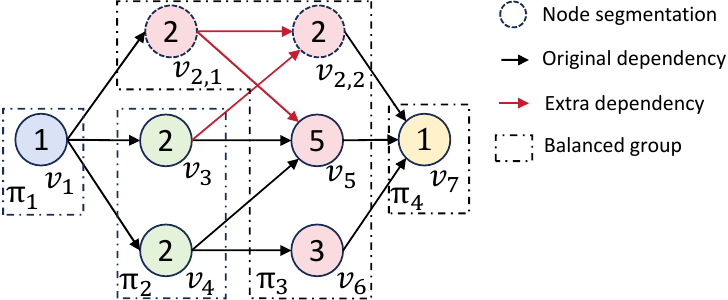}
    \caption{Example of proposed schedule}
    \label{fig:sche}
    \vspace{-0.2cm}
\end{figure}

\textbf{Example.} Figure~\ref{fig:sche} illustrates the proposed scheduling scheme for the example DAG in Figure~\ref{fig:dag_gpu} on a GPU with $M=6$ SMs.  
The balanced group $\pi_2=\{v_3, v_4\}$, configured with $m_3 = m_4 = 2$, does not fully occupy the GPU. Thus, the only node in $para(\pi_2)=\{v_2\}$ is selected to launch.
Since $\hat{C}_2 = 4$ exceeds the available capacity ($WS_2 = 2$), node $v_2$ is split into a parallel segment $v_{2,1}$ and a residual segment $v_{2,2}$, each with load $\hat{c}_{2,1}=\hat{c}_{2,2}=2$. The segment $v_{2,1}$ executes in parallel with $\pi_2$ under $m_{2,1}=2$, while the residual segment $v_{2,2}$ is scheduled with its own group $\pi_3$.
The extra dependencies $\{v_{2,1}, v_5\}$ and $\{v_3, v_{2,2}\}$ enforce the SM occupancy of each group and the sequential execution between $\pi_2$ and $\pi_3$.

\subsection{Response Time Analysis}\label{sec:ana}
The extra dependencies introduced by scheduling enforce a determined group execution order, enabling a predictable upper bound on the DAG response time.

\begin{lemma}\label{lemma:group_rt}
For any balanced group $\pi_j \in \Pi$, its relative response time $R(\pi_j)$, defined as the interval from the completion of $\pi_{j-1}$ to the completion of $\pi_j$, is upper-bounded as:
\begin{equation}\label{eq:R_pi_lemma}
    R(\pi_j) \le \max_{v_i \in \pi_j} C_i .
\end{equation}
\end{lemma}

\begin{proof}
The extra dependencies ensure that the nodes in $\pi_j$ execute concurrently.  
The parallelism scaling and node segmentation mechanisms avoid unpredictable delay.
Thus, $R(\pi_j)$ is bounded by its maximum node execution time.
\end{proof}

\begin{theorem}\label{theorem:dag_makespan}
The makespan of the DAG is upper-bounded as:
\begin{equation}\label{eq:R_G_theorem}
    R(G) \le \sum_{\pi_j \in \Pi} R(\pi_j).
\end{equation}
\end{theorem}

\begin{proof}
Since the proposed extra dependencies enforce a strictly sequential execution order over the groups in $\Pi$, the total makespan is bounded by the sum of all the relative response times. 
\end{proof}

\textbf{Sustainability:}
If the actual execution time of any node $v_i$ is less than its analytical bound $C_i$ in Eq. \eqref{eq:kernel_execution}, the corresponding group’s response time will not exceed the worst-case bound in Eq.~\eqref{eq:R_pi_lemma}. Such early completion also does not violate the sequential execution order enforced by the extra dependencies.  
Therefore, the timing analysis in Eq.~\eqref{eq:R_G_theorem} remains sound under early completion.

\textbf{Example.} Figure~\ref{fig:ana} illustrates the timing analysis, where the size of each kernel on SM represents its computation load.  
Figure~\ref{fig:ana}(a) shows the analysis under the Greedy schedule, where each $v_i$ is assigned $m_i^{\max}$. The concurrent node sets $\{v_2, v_3, v_4\}$ and $\{v_5, v_6\}$ require more SMs than the GPU can provide, leading to resource contention and unpredictable kernel execution delays.
As a comparison, Figure~\ref{fig:ana}(b) shows the analysis under the proposed scheduling. The proposed parallelism scaling and node segmentation ensure that the SM demands of $\pi_2$ and $\pi_3$ never exceed the GPU capacity, eliminating contention-induced delays. The extra dependency enforces sequential execution between $\pi_2$ and $\pi_3$, yielding a predictable makespan bounded by $R(G) \le 5$. Notably, this bound does not make any assumptions about the hardware scheduling.

\begin{figure}[t]
    \centering
    \includegraphics[width=1.0\linewidth]{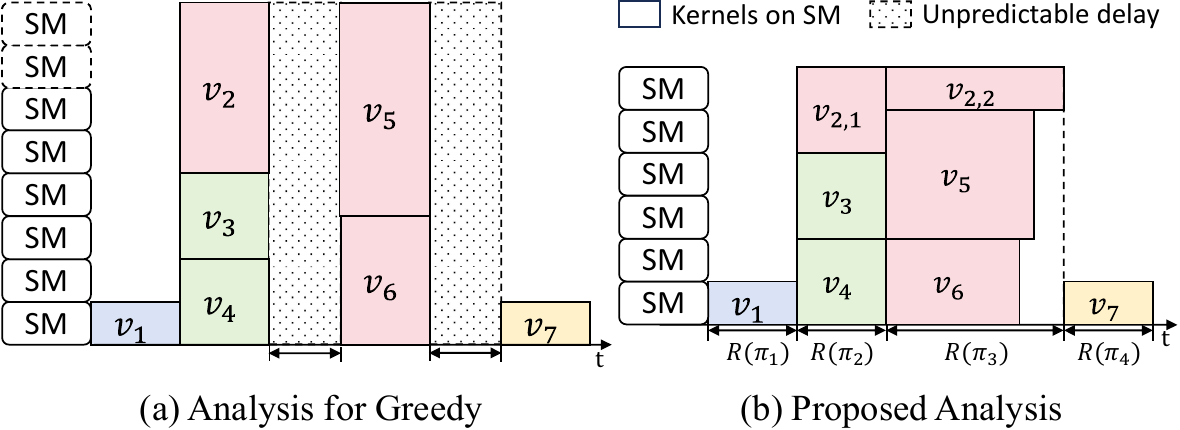}
    \caption{Example of timing analysis}
    \label{fig:ana}
    \vspace{-0.4cm}
\end{figure}

\section{Evaluation}\label{sec:evaluate}

\subsection{Experimental Setup}

To evaluate the worst-case makespan, a large set of synthesized DAGs was generated using the tool in~\cite{dag_gen}. Unless otherwise stated, the maximum DAG depth (\ie, number of layers) is randomly selected between 5 and 8. DAG construction starts from a single source node and proceeds layer by layer, with the number of nodes in each layer uniformly sampled between 2 and the maximum parallelism parameter $P$. For each setup, $1,000$ DAGs are randomly generated. Within each DAG, the node’s computation load is assigned uniformly, satisfying a given average load  $\hat{C}_{\text{avg}}$.

To evaluate the GPU task execution time, a case study on real-world benchmarks is performed. The proposed scheduling scheme and the competitive method are implemented using the CUDA Graph API, reducing the overhead of repeated kernel launches~\cite {NVIDIA_GettingStartedCudaGraphs_2020}. The required parallelism $m_i$ for each node is realized by configuring the kernel launch parameters $D_g$ and $D_b$ on each target platform. To emulate different computation loads across nodes, a single kernel function is used with various iteration counts.

\vspace{-0.2cm}
\subsection{Large-scale Experiment}

In this experiment, we compare the derived worst-case makespan bound in Eq.~\eqref{eq:R_G_theorem}
(\textbf{Proposed}) against $3$ representative methods: 

\textbf{Greedy.} Each kernel is assigned the maximum parallelism as $m_i = \min\{m_i^{max}, M\}$ for each $v_i$. As the kernel execution order and delay caused by resource competition are unpredictable, a sequential node execution is assumed in the analysis. 

\textbf{Greedy + unaware (Reference).} The parallelism assignment ignore the platform structure as $m_i = m_i^{max}$ for each $v_i$, with a sequential node execution order. The resulting makespan is used as a reference for normalization in analysis.

\textbf{Graham\_para.}
A parallel adaptation of Graham’s bound \cite{graham}, which analyzes the DAG's makespan without assuming node-level priorities. Each node is replaced with a set of parallel unit nodes with basic computation load ($\hat{C}=1$), ensuring each unit occupies exactly one SM, and the Graham bound is then applied to the transformed DAG.


\begin{figure}[ht]
    \centering
    \includegraphics[scale=0.42]{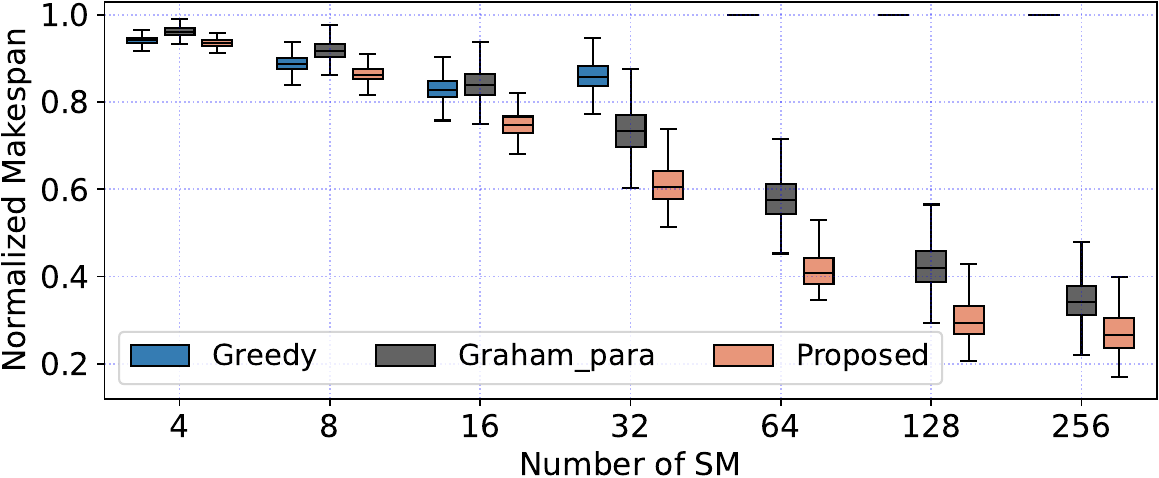}
    \caption{Makespan with various $M$ when $\hat{C}_{\text{avg}} = 20$}
    \label{fig:makespan_m}
    \vspace{-0.2cm}
\end{figure}

Figure~\ref{fig:makespan_m} reports the normalized makespan under various numbers of SMs, with $M$ ranging logarithmically from $4$ to $256$. The performance gap between the proposed method and the others generally widens as $M$ increases. When $M=128$, the proposed method outperforms Graham\_para by $28.7\%$ on average due to finer-grained exploitation of kernel concurrency; this gap narrows to $21.2\%$ when $M=256$, where both methods benefit from abundant SM resources and reach near-maximal parallelism. For $M \le 8$, the performance gap between Greedy and the proposed method remains below $3\%$, as kernels frequently saturate the GPU and execute mostly sequentially. As $M$ grows beyond $32$, Greedy fails to capitalize on the additional SMs, leading to noticeable degradation. Greedy and Greedy\_unaware converge when $M \ge 64$, where all kernels can achieve their maximum parallelism.

\begin{figure}[h]
    \centering
    \includegraphics[scale=0.42]{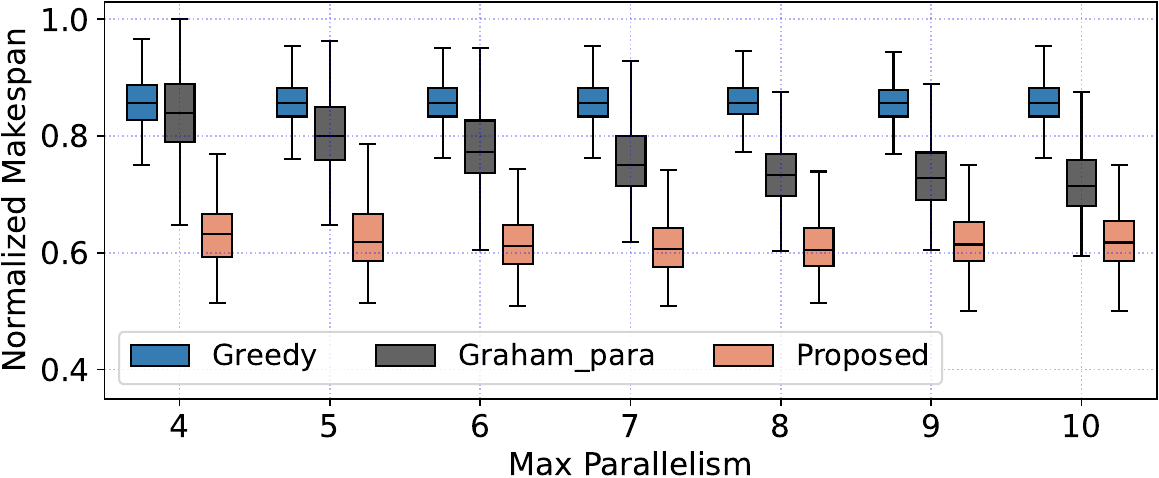}
    \caption{Makespan with various $P$ when $M=32,\hat{C}_{\text{avg}} = 20$}
    \label{fig:makespan_p}
    \vspace{-0.2cm}
\end{figure}

Figure~\ref{fig:makespan_p} reports the normalized makespan under different $P$. The proposed method consistently achieves the lowest makespan across all configurations, outperforming Graham\_para and Greedy by $18.9\%$ and $27.2\%$ on average, demonstrating its robustness on various DAG structures.
The performance gap between the proposed method and Greedy decreases from $23.9\%$ when $P=4$ to $13.7\%$ when $P=10$. This is because with a fixed $M$, increasing $P$ makes the DAG more likely to fully occupy the GPU and thereby reduces the advantage of the proposed parallelism scaling.

\begin{figure}[h]
    \centering
    \includegraphics[scale=0.42]{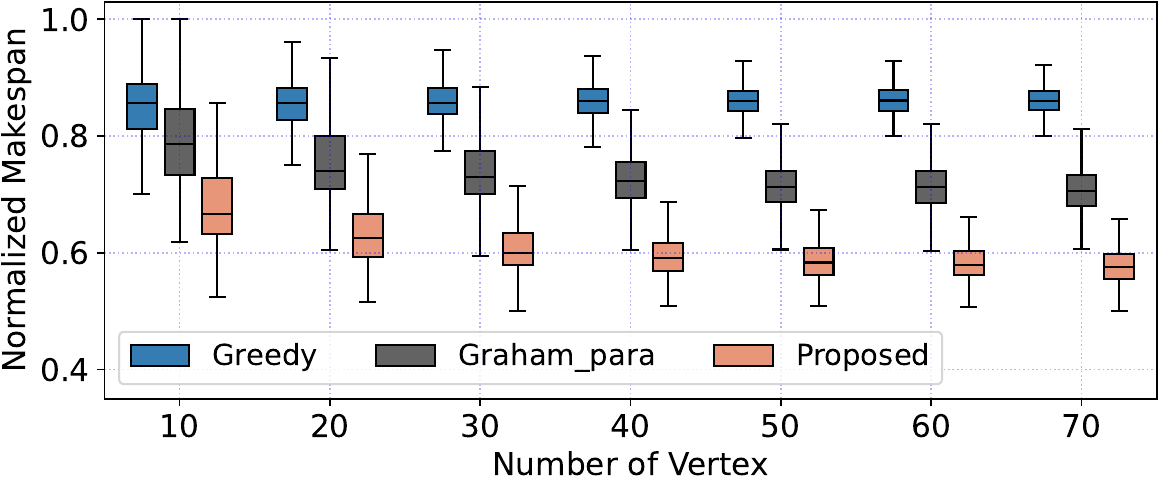}
    \caption{Makespan with various $|V|$ when $M=32, \hat{C}_{\text{avg}} = 20$}
    \label{fig:makespan_v}
    \vspace{-0.2cm}
\end{figure}

Figure~\ref{fig:makespan_v} illustrates the normalized makespan under various $|V|$. In this setup, $|V|$ is scaled by increasing the DAG depth while fixing $P=32$.
Across all configurations, the proposed method consistently yields the lowest makespan, showing robustness with the DAG size expansion.
The advantage of the proposed method increases as $|V|$ increases, up to $32.8\%$ and $18.2\%$ than Greedy and Graham\_para, respectively.
This is because a deeper DAG depth results in more balanced groups, amplifying the benefits of the proposed method.

\vspace{-0.1cm}

\subsection{Case Study}

A case study was conducted on $3$ DAG-structured benchmark tasks: Laplace, Gaussian elimination, and Stencil~\cite{heft,benoit2009contention}, on two GPU platforms: an NVIDIA RTX 3060 ($M=30$) and a Jetson Orin Nano ($M=8$). Each benchmark was executed $100$ times, with the measured task execution time summarized in Table~\ref{tab:case_study_c4} and~\ref{tab:case_study_c20}. The Greedy approach assigns maximum parallelism to each node without enforcing sequential constraints among nodes.

In most setups, the proposed method achieves more stable task execution times (lower standard deviation) due to its enforced execution order. It also achieves a lower execution time in most cases, with improvements up to an average $21.3\%$ when $M=8,\hat{C}_{\text{avg}}=4$ and $15.0\%$ when $M=30, \hat{C}_{\text{avg}}=20$ for Gaussian elimination. In these system setups, most kernels can neither saturate the GPU nor get its maximum parallelism, amplifying the benefit of the proposed parallelism scaling mechanism. When the average load is too small or too large for $M$, kernels either mostly reach maximum parallelism or fully occupy the GPU, reducing the proposed benefit, which is consistent with trends against Graham\_para in Figure~\ref{fig:makespan_m}.

\begin{table}[]
\centering
\caption{Task execution time of benchmarks when $\hat{C}_{avg}=4$}
\label{tab:case_study_c4}

\resizebox{\columnwidth}{!}{%
\begin{tabular}{|ll|lll|lll|}
\hline
\multicolumn{2}{|l|}{Method} & \multicolumn{3}{l|}{Proposed} & \multicolumn{3}{l|}{Greedy} \\ \hline
\multicolumn{2}{|l|}{Execution time(ms)} & \multicolumn{1}{l|}{max.} & \multicolumn{1}{l|}{\textbf{avg.}} & std. & \multicolumn{1}{l|}{max.} & \multicolumn{1}{l|}{avg.} & std. \\ \hline
\multicolumn{1}{|l|}{\multirow{3}{*}{M=8}} & Gaussian & \multicolumn{1}{l|}{74.25} & \multicolumn{1}{l|}{\textbf{51.23}} & 7.47 & \multicolumn{1}{l|}{85.61} & \multicolumn{1}{l|}{65.07} & 9.12 \\ \cline{2-8} 
\multicolumn{1}{|l|}{} & Laplace & \multicolumn{1}{l|}{100.83} & \multicolumn{1}{l|}{\textbf{83.30}} & 7.20 & \multicolumn{1}{l|}{121.72} & \multicolumn{1}{l|}{104.39} & 9.02 \\ \cline{2-8} 
\multicolumn{1}{|l|}{} & Stencil & \multicolumn{1}{l|}{114.12} & \multicolumn{1}{l|}{\textbf{86.07}} & 9.13 & \multicolumn{1}{l|}{122.76} & \multicolumn{1}{l|}{106.12} & 8.17 \\ \hline
\multicolumn{1}{|l|}{\multirow{3}{*}{M=30}} & Gaussian & \multicolumn{1}{l|}{10.22} & \multicolumn{1}{l|}{\textbf{7.32}} & 0.33 & \multicolumn{1}{l|}{12.25} & \multicolumn{1}{l|}{7.38} & 0.60 \\ \cline{2-8} 
\multicolumn{1}{|l|}{} & Laplace & \multicolumn{1}{l|}{11.32} & \multicolumn{1}{l|}{\textbf{11.05}} & 0.14 & \multicolumn{1}{l|}{11.35} & \multicolumn{1}{l|}{11.08} & 0.15 \\ \cline{2-8} 
\multicolumn{1}{|l|}{} & Stencil & \multicolumn{1}{l|}{11.98} & \multicolumn{1}{l|}{\textbf{11.37}} & 0.15 & \multicolumn{1}{l|}{11.51} & \multicolumn{1}{l|}{11.09} & 0.17 \\ \hline
\end{tabular}%
}
\end{table}

\begin{table}[]
\centering
\caption{Task execution time of benchmarks when $\hat{C}_{avg}=20$}
\label{tab:case_study_c20}

\resizebox{\columnwidth}{!}{%
\begin{tabular}{|ll|lll|lll|}
\hline
\multicolumn{2}{|l|}{Method} & \multicolumn{3}{l|}{Proposed} & \multicolumn{3}{l|}{Greedy} \\ \hline
\multicolumn{2}{|l|}{Execution time(ms)} & \multicolumn{1}{l|}{max.} & \multicolumn{1}{l|}{\textbf{avg.}} & std. & \multicolumn{1}{l|}{max.} & \multicolumn{1}{l|}{avg.} & std. \\ \hline
\multicolumn{1}{|l|}{\multirow{3}{*}{M=8}} & Gaussian & \multicolumn{1}{l|}{193.11} & \multicolumn{1}{l|}{\textbf{165.24}} & 11.36 & \multicolumn{1}{l|}{194.50} & \multicolumn{1}{l|}{167.90} & 12.34 \\ \cline{2-8} 
\multicolumn{1}{|l|}{} & Laplace & \multicolumn{1}{l|}{309.29} & \multicolumn{1}{l|}{\textbf{280.31}} & 12.17 & \multicolumn{1}{l|}{331.42} & \multicolumn{1}{l|}{306.65} & 10.36 \\ \cline{2-8} 
\multicolumn{1}{|l|}{} & Stencil & \multicolumn{1}{l|}{333.62} & \multicolumn{1}{l|}{\textbf{297.03}} & 12.12 & \multicolumn{1}{l|}{353.54} & \multicolumn{1}{l|}{324.89} & 12.87 \\ \hline
\multicolumn{1}{|l|}{\multirow{3}{*}{M=30}} & Gaussian & \multicolumn{1}{l|}{29.05} & \multicolumn{1}{l|}{\textbf{23.96}} & 0.99 & \multicolumn{1}{l|}{31.71} & \multicolumn{1}{l|}{28.20} & 1.05 \\ \cline{2-8} 
\multicolumn{1}{|l|}{} & Laplace & \multicolumn{1}{l|}{46.81} & \multicolumn{1}{l|}{\textbf{40.74}} & 1.43 & \multicolumn{1}{l|}{51.30} & \multicolumn{1}{l|}{44.82} & 1.43 \\ \cline{2-8} 
\multicolumn{1}{|l|}{} & Stencil & \multicolumn{1}{l|}{49.44} & \multicolumn{1}{l|}{\textbf{43.39}} & 3.11 & \multicolumn{1}{l|}{54.00} & \multicolumn{1}{l|}{48.45} & 3.39 \\ \hline
\end{tabular}%
}
\end{table}
\vspace{-0.2cm}

\section{Conclusion}
This paper presented a scheduling and timing analysis framework for DAG-structured GPU tasks. By decomposing the DAG into balanced groups, the method enables fine-grained scheduling and analysis, providing a reduced and predictable makespan. The scheduling is fully realizable within the scope of the standard CUDA API, without additional hardware or software support. The corresponding timing analysis does not make any assumption about kernel-level priority. Extensive experiments on large-scale synthetic DAGs and a case study prove the effectiveness of the proposed approach. Future work will extend the framework to conditional DAGs on heterogeneous platforms.
\bibliographystyle{ACM-Reference-Format}
\bibliography{dag_gpu}


\end{document}